\documentclass[11pt]{article}

\usepackage[T1]{fontenc}
\usepackage{amsmath,amssymb,amsthm,mathtools}
\usepackage[margin=1in]{geometry}
\usepackage{microtype}
\usepackage[hidelinks]{hyperref}
\hypersetup{
  pdftitle={Unitary-orbit classification and a refinement theorem for context-independent projective probabilities},
  pdfauthor={Michael P. Rubin},
  pdfkeywords={Gleason theorem, projective measurement, context independence, refinement consistency, unitary orbit, refinement graph}
}

\newtheorem{theorem}{Theorem}[section]
\newtheorem{proposition}[theorem]{Proposition}
\newtheorem{corollary}[theorem]{Corollary}
\newtheorem{lemma}[theorem]{Lemma}
\theoremstyle{definition}
\newtheorem{definition}[theorem]{Definition}
\newtheorem{example}[theorem]{Example}
\theoremstyle{remark}
\newtheorem{remark}[theorem]{Remark}

\DeclareMathOperator{\rank}{rank}
\DeclareMathOperator{\Tr}{Tr}
\DeclareMathOperator{\ran}{ran}
\newcommand{\HH}{\mathcal H}
\newcommand{\UU}{\mathcal U}
\newcommand{\PP}{\mathcal P}
\newcommand{\MM}{\mathfrak M}
\newcommand{\id}{I}
\newcommand{\C}{\mathbb C}
\newcommand{\abs}[1]{\left|#1\right|}
\newcommand{\set}[1]{\left\{#1\right\}}

\title{Unitary-orbit classification and a refinement theorem for\\context-independent projective probabilities}
\author{Michael P. Rubin}
\date{August 2, 2026}

\begin{document}
\maketitle

\begin{abstract}
Let $\HH$ be a finite-dimensional complex Hilbert space, and let $w(P,\mathsf M)$ be a normalized probability weight assigned to an outcome projection $P$ as it occurs in a projective measurement $\mathsf M$. For fixed $P$, let $G_P\cong\UU(P^\perp)$ be the group of unitaries acting identically on $\ran P$ and arbitrarily on $P^\perp$. We classify the $G_P$-orbits of projective measurements containing $P$: two measurements lie in the same orbit exactly when the multisets of ranks of their complementary outcomes agree. The orbit with profile $\lambda=(r_1,\ldots,r_k)$ is a compact homogeneous space of real dimension $(d-\rank P)^2-\sum_j r_j^2$. On maximal rank-one measurements there is one orbit, so context independence is equivalent to $G_P$-invariance, with equality of the corresponding uniform defects; invariance under two-level complementary unitaries already suffices. For arbitrary projective measurements, every context containing an outcome $P\neq I$ coarsens to the unique binary context $\{P,I-P\}$. Consequently, refinement consistency alone is equivalent to context independence. The finite orbit space carries a natural rank-profile refinement graph, whose diameter is $d-\rank P-1$ when $P\neq I$. We prove a stability theorem that compares the binary-coarsening path with a shortest path in this graph followed by one complementary unitary. In dimension at least three, Gleason's theorem converts the maximal-context invariance condition and the all-context refinement condition, on their respective domains, into the Born form $\Tr(\rho P)$. The results are structural characterizations, not independent physical derivations of context independence.
\end{abstract}


\section{Introduction}

Let $\HH$ be a complex Hilbert space of finite dimension $d\geq 2$. In one finite-dimensional formulation, Gleason's theorem begins with a function on rank-one projections whose values sum to one on every orthonormal basis and concludes, for $d\geq 3$, that the function is given by the Born rule for a unique density operator~\cite{Gleason1957}. Because the function is defined on projections rather than on projection--measurement pairs, the same projection carries the same value in every maximal projective measurement in which it occurs. In this precise sense, context independence is built into the domain of a frame function.

To display that condition explicitly, one may instead assign a normalized weight $w(P,\mathsf M)$ to the pair consisting of an outcome projection $P$ and a projective measurement $\mathsf M$ containing it. The question considered here is precise: which transformations of the surrounding measurement leave the weight of $P$ unchanged, and when do those invariances amount to full context independence?

For a fixed outcome $P$, changing only the remaining outcomes amounts to changing an orthogonal decomposition of $P^\perp$. The relevant symmetry group is therefore
\[
G_P=\set{U\in\UU(\HH):Ux=x\text{ for every }x\in\ran P}\cong\UU(P^\perp).
\]
The first result is an orbit classification. Two projective measurements containing $P$ are related by $G_P$ exactly when the ranks of the other outcome projections agree, including multiplicity. Thus the $G_P$-orbits are indexed by integer partitions of $d-\rank P$, and each orbit has an explicit compact homogeneous-space form. In particular, all maximal rank-one measurements containing a fixed rank-one $P$ belong to one orbit. On that domain, complementary-unitary invariance and context independence are therefore equivalent. The corresponding contextual oscillation and invariance defect are exactly equal. A local form of the result shows that it is enough to test unitaries supported on two-dimensional subspaces of $P^\perp$.

For arbitrary projective measurements, complementary-unitary invariance is weaker than context independence because the group action cannot change the complementary rank profile. A normalized assignment that gives every outcome the uniform weight $1/|\mathsf M|$ is invariant under the full unitary group but depends on the number of outcomes. Refinements provide the missing links, but in a stronger way than a profile-by-profile argument alone suggests. Every measurement containing a proper projection $P$ can be coarsened, by merging complementary outcomes, to the unique binary measurement $\{P,I-P\}$. It follows that refinement consistency by itself is equivalent to context independence on the full class of projective measurements.

The quantitative statement retains a useful role for unitary invariance. The quotient $\MM_P/G_P$ is the finite set of integer partitions of $d-\rank P$, and elementary complementary refinements make this quotient into a natural graph. One route between two contexts coarsens both to the unique binary context and uses only refinement control. A second route follows a shortest path between their rank profiles and then uses one complementary unitary within the terminal orbit. The resulting pairwise estimate is sharper than forcing both contexts through the maximal rank-one profile.

Combining these characterizations with Gleason's theorem yields two Born-form statements. On maximal rank-one measurements, complementary-unitary invariance is equivalent to the Born form. On all projective measurements, refinement consistency is equivalent to the Born form. Neither statement supplies an independent physical justification for the relevant condition; each identifies the exact mathematical content of context independence on its stated domain.

The transitivity of a unitary group on orthogonal decompositions with prescribed dimensions is standard. The contribution here is to organize the context space into complete rank-profile orbits, identify their compact homogeneous-space geometry, equip the finite orbit quotient with its refinement graph, and derive exact and quantitative consequences for context-indexed probability assignments. Rank data have also been used to partition projective measurements in the study of projective simulability of positive-operator-valued measurements~\cite{CobucciEtAl2025}; here the complementary rank multiset is identified as the complete invariant of the subgroup fixing a specified outcome. Contextual formulations of Gleason's theorem and related structural results have been developed in several settings; see, for example, D\"oring and Frembs~\cite{DoringFrembs2022}. Consistency of probability assignments across families of bases was studied by Patra and van der Meyden~\cite{PatraMeyden2012}, and the wider operational notion of generalized contextuality was formulated by Spekkens~\cite{Spekkens2005}. The present note concerns only scalar probability assignments on finite-dimensional projective measurements and is not a no-go theorem for ontological models.

Theorem~\ref{thm:orbit} gives the orbit classification, Proposition~\ref{prop:homogeneous} its homogeneous-space geometry, Theorem~\ref{thm:profile} and Proposition~\ref{prop:defect} the exact orbitwise characterizations, and Theorem~\ref{thm:twolevel} the elementary-unitary reduction. Section~\ref{sec:refinement} contains the exact refinement theorem, the rank-profile graph, and the orbit--refinement stability theorem. Theorems~\ref{thm:gleason} and~\ref{thm:allgleason} give the Gleason consequences.

\section{Projective measurements and complementary unitaries}

Write $\PP(\HH)$ for the set of orthogonal projections on $\HH$ and $\UU(\HH)$ for its unitary group. We use $P^\perp$ for the orthogonal complement $\ran(\id-P)$.

\begin{definition}[Projective measurement]
A projective measurement is an unordered finite family
\[
\mathsf M=\set{P_1,\ldots,P_m}\subseteq\PP(\HH)
\]
of nonzero pairwise orthogonal projections such that $\sum_{j=1}^m P_j=\id$. Let $\MM(\HH)$ denote the set of all projective measurements. A measurement is \emph{maximal} or \emph{fine-grained} if every outcome has rank one; the set of maximal measurements is denoted by $\MM_1(\HH)$.
\end{definition}

For $P\in\PP(\HH)$, define
\[
\MM_P=\set{\mathsf M\in\MM(\HH):P\in\mathsf M}.
\]
The full unitary stabilizer of $P$ is
\[
\operatorname{Stab}(P)=\set{U\in\UU(\HH):UPU^*=P}.
\]
With respect to $\HH=\ran P\oplus P^\perp$, it is isomorphic to $\UU(\ran P)\times\UU(P^\perp)$. Its action by conjugation on $\MM_P$ factors through the second component: the first component acts only inside $\ran P$ and leaves every projection in a measurement containing $P$ unchanged. It is therefore convenient to use the following effective subgroup.

\begin{definition}[Complementary unitary group]
For $P\in\PP(\HH)$, let
\[
G_P=\set{U\in\UU(\HH):U|_{\ran P}=\id_{\ran P}}.
\]
Every $U\in G_P$ leaves $P^\perp$ invariant, and restriction to $P^\perp$ identifies $G_P$ with $\UU(P^\perp)$. The group acts on $\MM_P$ by
\[
U\cdot\mathsf M=U\mathsf M U^*:=\set{UQU^*:Q\in\mathsf M}.
\]
\end{definition}

\begin{definition}[Complementary rank profile]
Let $P\in\mathsf M$. The \emph{complementary rank profile} of $P$ in $\mathsf M$, denoted $\lambda_P(\mathsf M)$, is the nondecreasing tuple
\[
\lambda_P(\mathsf M)=(r_1,\ldots,r_k)
\]
obtained by listing the ranks of the projections in $\mathsf M\setminus\set{P}$. Thus each $r_j$ is a positive integer and
\[
\sum_{j=1}^k r_j=d-\rank P.
\]
For a profile $\lambda$, let
\[
\MM_{P,\lambda}=\set{\mathsf M\in\MM_P:\lambda_P(\mathsf M)=\lambda}.
\]
\end{definition}

The profile is an integer partition of $d-\rank P$. Because measurements are treated as unordered families, only the multiset of complementary ranks matters.

\section{Orbit classification}

\begin{theorem}[Orbit classification]\label{thm:orbit}
Let $P$ be a nonzero projection and let $\mathsf M,\mathsf N\in\MM_P$. Then the following are equivalent:
\begin{enumerate}
\item there exists $U\in G_P$ such that $U\mathsf M U^*=\mathsf N$;
\item $\lambda_P(\mathsf M)=\lambda_P(\mathsf N)$.
\end{enumerate}
Consequently, every nonempty set $\MM_{P,\lambda}$ is a single $G_P$-orbit, and the orbits of $G_P$ on $\MM_P$ are indexed by the integer partitions of $d-\rank P$.
\end{theorem}

\begin{proof}
Suppose first that $U\mathsf M U^*=\mathsf N$ for some $U\in G_P$. Unitary conjugation preserves rank and fixes $P$, so it carries the complementary outcomes of $\mathsf M$ bijectively to those of $\mathsf N$ without changing their ranks. Hence the complementary rank profiles agree.

Conversely, assume $\lambda_P(\mathsf M)=\lambda_P(\mathsf N)$. Write
\[
\mathsf M\setminus\set{P}=\set{Q_1,\ldots,Q_k},\qquad
\mathsf N\setminus\set{P}=\set{R_1,\ldots,R_k},
\]
where the outcomes have been indexed so that $\rank Q_j=\rank R_j$ for every $j$. Choose a unitary map
\[
V_j:\ran Q_j\longrightarrow\ran R_j
\]
for each $j$. Since the ranges of the $Q_j$ form an orthogonal direct-sum decomposition of $P^\perp$, as do the ranges of the $R_j$, the direct sum
\[
V=\bigoplus_{j=1}^k V_j:P^\perp\longrightarrow P^\perp
\]
is unitary. Define $U=\id_{\ran P}\oplus V$. Then $U\in G_P$, $UQ_jU^*=R_j$ for every $j$, and therefore $U\mathsf M U^*=\mathsf N$.
\end{proof}

\begin{proposition}[Homogeneous-space geometry]\label{prop:homogeneous}
Let $n=d-\rank P$ and let $\lambda=(r_1,\ldots,r_k)$ be a complementary rank profile. If $n=0$, then $P=I$ and $\MM_{P,\lambda}$ is a point. Suppose $n\geq1$. For each $s\geq1$, let $m_s$ be the multiplicity of the part $s$ in $\lambda$. Noncanonically,
\[
\MM_{P,\lambda}\cong
\UU(n)\Big/\prod_{s:m_s>0}\left(\UU(s)^{m_s}\rtimes S_{m_s}\right).
\]
In particular, $\MM_{P,\lambda}$ is a compact smooth homogeneous space with
\[
\dim_{\mathbb R}\MM_{P,\lambda}
=n^2-\sum_{s:m_s>0}m_s s^2
=n^2-\sum_{j=1}^k r_j^2.
\]
Thus $\MM_P$ is a finite disjoint union of compact homogeneous spaces indexed by the integer partitions of $n$.
\end{proposition}

\begin{proof}
The case $n=0$ is immediate. Assume $n\geq1$, and choose a reference orthogonal decomposition of $P^\perp$ with dimension profile $\lambda$. The group $\UU(n)$ acts transitively on decompositions with that profile by Theorem~\ref{thm:orbit}. The setwise stabilizer of the reference decomposition consists of arbitrary unitary changes of basis within each block, together with permutations of blocks having the same dimension. This stabilizer is the displayed product of wreath products; it is a closed subgroup of $\UU(n)$, being the isotropy group of a point under a continuous action. The quotient therefore carries the standard smooth manifold structure for which the projection $\UU(n)\to \UU(n)/H$ is a submersion. The quotient description follows from the orbit--stabilizer theorem for compact Lie-group actions. The finite permutation factors do not affect dimension, while $\dim_{\mathbb R}\UU(q)=q^2$, giving the stated formula.
\end{proof}

\begin{corollary}[Transitivity on maximal contexts]\label{cor:transitive}
Let $P$ be rank one. The action of $G_P$ on
\[
\MM_{1,P}:=\MM_1(\HH)\cap\MM_P
\]
is transitive.
\end{corollary}

\begin{proof}
Every maximal measurement containing $P$ has complementary rank profile
\[
(\underbrace{1,\ldots,1}_{d-1\text{ times}}),
\]
so the assertion follows from Theorem~\ref{thm:orbit}.
\end{proof}

For $d=2$, the set $\MM_{1,P}$ consists of the single measurement $\set{P,\id-P}$. For $d\geq3$, it is a nontrivial homogeneous space, but it remains one orbit.

\section{Invariant context-indexed weights}

\begin{definition}[Context-indexed probability assignment]
A context-indexed probability assignment on a class $\mathcal D\subseteq\MM(\HH)$ is a function
\[
w:\set{(P,\mathsf M):\mathsf M\in\mathcal D,\ P\in\mathsf M}\longrightarrow[0,1]
\]
such that
\[
\sum_{P\in\mathsf M}w(P,\mathsf M)=1
\qquad\text{for every }\mathsf M\in\mathcal D.
\]
It is \emph{context independent} if
\[
w(P,\mathsf M)=w(P,\mathsf N)
\]
whenever $P\in\mathsf M\cap\mathsf N$.
\end{definition}

\begin{definition}[Complementary-unitary invariance]
A context-indexed probability assignment on $\MM(\HH)$ is complementary-unitary invariant if
\[
w(P,U\mathsf M U^*)=w(P,\mathsf M)
\]
for every $P\in\mathsf M$ and every $U\in G_P$. The same definition applies on maximal measurements, with $P$ restricted to rank-one outcomes and $\mathsf M\in\MM_1(\HH)$.
\end{definition}

The term \emph{invariance} is used deliberately. The assignment $w$ and the outcome $P$ are held fixed while only the complementary measurement is transformed.

\begin{theorem}[Profile characterization]\label{thm:profile}
A context-indexed probability assignment $w$ on $\MM(\HH)$ is complementary-unitary invariant if and only if, for every nonzero projection $P$, there is a function $F_P$ on the integer partitions of $d-\rank P$ such that
\[
w(P,\mathsf M)=F_P\bigl(\lambda_P(\mathsf M)\bigr)
\qquad(\mathsf M\in\MM_P).
\]
For an invariant assignment represented in this way, context independence holds if and only if each $F_P$ is constant on the profiles that occur. In particular, on maximal measurements, complementary-unitary invariance is equivalent to context independence.
\end{theorem}

\begin{proof}
By Theorem~\ref{thm:orbit}, the sets $\MM_{P,\lambda}$ are exactly the $G_P$-orbits in $\MM_P$. Thus $w(P,\cdot)$ is $G_P$-invariant precisely when it is constant on each $\MM_{P,\lambda}$, which is equivalent to the stated factorization through the profile. Context independence requires the same value across all profiles. On maximal measurements containing a fixed rank-one $P$, Corollary~\ref{cor:transitive} gives only one orbit and one profile.
\end{proof}

\begin{example}[Why coarse graining matters]\label{ex:uniform}
Define, on all projective measurements,
\[
w_{\mathrm{unif}}(P,\mathsf M)=\frac{1}{|\mathsf M|}.
\]
This assignment is normalized and invariant under the full unitary group, since unitary conjugation preserves the number of outcomes. It is nevertheless context dependent when $d\geq3$. If $P$ is rank one, then
\[
w_{\mathrm{unif}}\bigl(P,\set{P,\id-P}\bigr)=\frac12,
\]
whereas a maximal measurement $\mathsf M\ni P$ has $d$ outcomes and
\[
w_{\mathrm{unif}}(P,\mathsf M)=\frac1d.
\]
The two measurements belong to different rank-profile orbits. Thus complementary-unitary invariance controls variation within a profile but does not compare different coarse-graining profiles.
\end{example}

The equivalence on an orbit has an exact quantitative version. For any nonempty $\MM_{P,\lambda}$, define the contextual oscillation
\[
\operatorname{osc}_{P,\lambda}(w)
=\sup_{\mathsf M,\mathsf N\in\MM_{P,\lambda}}
\abs{w(P,\mathsf M)-w(P,\mathsf N)}
\]
and the complementary-unitary defect
\[
\operatorname{def}_{P,\lambda}(w)
=\sup_{\substack{\mathsf M\in\MM_{P,\lambda}\\U\in G_P}}
\abs{w(P,U\mathsf M U^*)-w(P,\mathsf M)}.
\]

\begin{proposition}[Exact defect identity]\label{prop:defect}
For every $P$, every nonempty profile class $\MM_{P,\lambda}$, and every context-indexed assignment $w$,
\[
\operatorname{osc}_{P,\lambda}(w)=\operatorname{def}_{P,\lambda}(w).
\]
Consequently, on maximal measurements, $w$ is $\varepsilon$-context independent at each fixed outcome if and only if it is $\varepsilon$-invariant under all complementary unitaries.
\end{proposition}

\begin{proof}
Every pair $(\mathsf M,U\mathsf M U^*)$ appearing in the definition of $\operatorname{def}_{P,\lambda}$ is a pair of measurements in $\MM_{P,\lambda}$, so
\[
\operatorname{def}_{P,\lambda}(w)\leq\operatorname{osc}_{P,\lambda}(w).
\]
Conversely, for every $\mathsf M,\mathsf N\in\MM_{P,\lambda}$, Theorem~\ref{thm:orbit} supplies $U\in G_P$ with $U\mathsf M U^*=\mathsf N$. Hence each difference contributing to the oscillation also contributes to the defect, proving the reverse inequality.
\end{proof}

For a rank-one projection $P$, call $U\in G_P$ \emph{two-level} if there is a subspace $L\subseteq P^\perp$ with $\dim L\leq2$ such that $U$ is the identity on $P^\perp\ominus L$. Let $G_P^{(2)}$ denote the set of such unitaries, and define on maximal measurements
\[
\operatorname{def}^{(2)}_P(w)
=\sup_{\substack{\mathsf M\in\MM_{1,P}\\U\in G_P^{(2)}}}
\abs{w(P,U\mathsf M U^*)-w(P,\mathsf M)}
\]
and
\[
\operatorname{osc}^{\max}_P(w)
=\sup_{\mathsf M,\mathsf N\in\MM_{1,P}}
\abs{w(P,\mathsf M)-w(P,\mathsf N)}.
\]

\begin{theorem}[Elementary-unitary reduction]\label{thm:twolevel}
Let $w$ be a context-indexed probability assignment on $\MM_1(\HH)$. Let $P$ be rank one and put $n=d-1$. Then
\[
\operatorname{osc}^{\max}_P(w)
\leq \binom{n}{2}\operatorname{def}^{(2)}_P(w).
\]
Therefore, on maximal measurements, invariance under all two-level complementary unitaries is equivalent to full complementary-unitary invariance and to context independence.
\end{theorem}

\begin{proof}
Fix $\mathsf M,\mathsf N\in\MM_{1,P}$, and choose orthonormal bases $e_1,\ldots,e_n$ and $f_1,\ldots,f_n$ of $P^\perp$ whose rank-one projections are the complementary outcomes of $\mathsf M$ and $\mathsf N$, respectively. Let $V\in\UU(P^\perp)$ satisfy $Ve_j=f_j$. A standard complex Givens reduction gives
\[
V=G_m\cdots G_1D,
\qquad m\leq \binom{n}{2},
\]
where each $G_j$ is two-level and $D$ is diagonal in the $e$-basis; see, for example,~\cite[Chapter~5]{GolubVanLoan2013} for Givens reduction and~\cite[Sec.~4.5.1]{NielsenChuang2010,ReckEtAl1994} for closely related two-level decompositions in quantum information and linear optics. The unitary $D$ fixes every complementary rank-one projection of $\mathsf M$. Extending the $G_j$ and $D$ by the identity on $\ran P$, the contexts $\mathsf M$ and $\mathsf N$ are therefore joined by at most $m$ two-level conjugations. The triangle inequality gives
\[
\abs{w(P,\mathsf M)-w(P,\mathsf N)}
\leq m\operatorname{def}^{(2)}_P(w)
\leq \binom{n}{2}\operatorname{def}^{(2)}_P(w).
\]
Taking the supremum proves the bound; no attempt is made to optimize the constant. If the two-level defect vanishes, the maximal contextual oscillation vanishes. The converse implications are immediate.
\end{proof}

\section{Refinement consistency and arbitrary contexts}\label{sec:refinement}

An elementary refinement splits one complementary outcome into two orthogonal outcomes while leaving the outcome of interest unchanged.

\begin{definition}[Elementary complementary refinement]
Let $\mathsf M$ be a projective measurement. A measurement $\mathsf M'$ is an \emph{elementary refinement} of $\mathsf M$ if there is an outcome $Q\in\mathsf M$ and nonzero orthogonal projections $Q_1,Q_2$ such that
\[
Q=Q_1+Q_2,
\qquad
\mathsf M'=\bigl(\mathsf M\setminus\set{Q}\bigr)\cup\set{Q_1,Q_2}.
\]
If $P\in\mathsf M\setminus\set{Q}$ is an unaffected outcome, we write $\mathsf M'\succ_P\mathsf M$ and call this an \emph{elementary complementary refinement relative to $P$}. A context-indexed assignment is \emph{refinement consistent} if
\[
w(P,\mathsf M')=w(P,\mathsf M)
\]
whenever $\mathsf M'\succ_P\mathsf M$.
\end{definition}

For normalized assignments, this local condition also gives the usual additivity under coarse graining.

\begin{lemma}[Refinement consistency and additivity]\label{lem:additivity}
Let $w$ be normalized and refinement consistent. In the setting of the preceding definition,
\[
w(Q,\mathsf M)=w(Q_1,\mathsf M')+w(Q_2,\mathsf M').
\]
More generally, the weight of a coarse outcome equals the sum of the weights of the outcomes in any finite refinement of it.
\end{lemma}

\begin{proof}
Every outcome of $\mathsf M$ other than $Q$ is unchanged in $\mathsf M'$, so refinement consistency identifies its two weights. Subtracting the two normalization identities leaves the displayed equality. Iteration gives the finite-refinement statement.
\end{proof}

Refinement consistency is vacuous on $\MM_1(\HH)$: an elementary refinement increases the number of outcomes, so no maximal measurement is an elementary refinement of another maximal measurement. The two hypotheses used in this paper therefore play different roles. Complementary-unitary invariance is the operative condition on maximal measurements, where the rank profile is fixed and refinements leave the domain, whereas refinement consistency is the operative condition on the full class $\MM(\HH)$, where coarse-graining profiles vary.

\begin{theorem}[Refinement theorem]\label{thm:refinement}
A context-indexed probability assignment on all projective measurements is context independent if and only if it is refinement consistent.
\end{theorem}

\begin{proof}
Context independence immediately implies refinement consistency. Conversely, fix an outcome projection $P$. If $P=I$, the only context is $\{I\}$. If $P\neq I$, every measurement $\mathsf M\in\MM_P$ can be coarsened to the binary measurement $\mathsf B_P=\{P,I-P\}$ by merging complementary outcomes. Each merge is the reverse of an elementary complementary refinement relative to $P$, so refinement consistency gives
\[
w(P,\mathsf M)=w(P,\mathsf B_P).
\]
The same equality holds for every $\mathsf N\in\MM_P$, and hence $w(P,\mathsf M)=w(P,\mathsf N)$. Thus $w$ is context independent.
\end{proof}

The exact theorem has a quantitative refinement that uses the finite orbit space itself.

\begin{definition}[Rank-profile refinement graph]\label{def:profilegraph}
For $n\geq0$, let $\Gamma_n$ be the finite graph whose vertices are the integer partitions of $n$, with the empty partition as the sole vertex when $n=0$. Two partitions are adjacent when one is obtained from the other by replacing one part $r$ by two positive parts $a,b$ with $a+b=r$. Let $d_n$ denote the graph distance in $\Gamma_n$, and let $\ell(\lambda)$ denote the number of parts of a partition $\lambda$.
\end{definition}

\begin{proposition}[Diameter of the rank-profile graph]\label{prop:graphdiameter}
For $n\geq1$ and every partition $\lambda$ of $n$,
\[
d_n\bigl(\lambda,(n)\bigr)=\ell(\lambda)-1,
\qquad
d_n\bigl(\lambda,(1^n)\bigr)=n-\ell(\lambda).
\]
Consequently,
\[
\operatorname{diam}(\Gamma_n)=n-1.
\]
For $n=0$, the graph consists of one vertex and has diameter zero.
\end{proposition}

\begin{proof}
Each edge changes the number of parts by exactly one. Merging parts until only one remains gives a path of length $\ell(\lambda)-1$ from $\lambda$ to $(n)$, and no shorter path is possible. Likewise, splitting every part into ones gives a path of length $n-\ell(\lambda)$ from $\lambda$ to $(1^n)$, and no shorter path is possible. Hence, for any partitions $\lambda,\mu$ of $n$,
\[
d_n(\lambda,\mu)
\leq
\min\bigl\{\ell(\lambda)+\ell(\mu)-2,\;2n-\ell(\lambda)-\ell(\mu)\bigr\}
\leq n-1.
\]
The two extreme partitions $(n)$ and $(1^n)$ are at distance $n-1$, so the diameter is exactly $n-1$.
\end{proof}

For $P\in\mathsf M$, set
\[
n_P=d-\rank P,
\qquad
k_P(\mathsf M)=\abs{\mathsf M\setminus\set{P}},
\]
and define
\[
c_P(\mathsf M)=\max\set{k_P(\mathsf M)-1,0},
\qquad
s_P(\mathsf M)=n_P-k_P(\mathsf M).
\]
If $P\neq I$, then $c_P(\mathsf M)=\ell(\lambda_P(\mathsf M))-1$ is the number of elementary merges required to coarsen the complementary outcomes to $I-P$, while $s_P(\mathsf M)=n_P-\ell(\lambda_P(\mathsf M))$ is the number of elementary splits required to refine them all to rank one. In that case,
\[
c_P(\mathsf M)+s_P(\mathsf M)=n_P-1.
\]
For $P=I$, both quantities are zero and $\mathsf M=\{I\}$ is the only context.

Define the global complementary-unitary and refinement defects at $P$ by
\[
\operatorname{def}^{G}_P(w)
=\sup_{\substack{\mathsf M\in\MM_P\\U\in G_P}}
\abs{w(P,U\mathsf M U^*)-w(P,\mathsf M)}
\]
and
\[
\operatorname{def}^{\mathrm{ref}}_P(w)
=\sup_{\mathsf M'\succ_P\mathsf M}
\abs{w(P,\mathsf M')-w(P,\mathsf M)},
\]
where the supremum over an empty family is defined to be zero. Also write
\[
\operatorname{osc}_P(w)
=\sup_{\mathsf M,\mathsf N\in\MM_P}
\abs{w(P,\mathsf M)-w(P,\mathsf N)}.
\]

\begin{theorem}[Orbit--refinement stability theorem]\label{thm:bridge}
Let $w$ be a context-indexed probability assignment on $\MM(\HH)$. For every nonzero projection $P$ and every $\mathsf M,\mathsf N\in\MM_P$,
\begin{align}
\abs{w(P,\mathsf M)-w(P,\mathsf N)}
\leq\min\Bigl\{&
\bigl(c_P(\mathsf M)+c_P(\mathsf N)\bigr)
\operatorname{def}^{\mathrm{ref}}_P(w),\nonumber\\
&d_{n_P}\bigl(\lambda_P(\mathsf M),\lambda_P(\mathsf N)\bigr)
\operatorname{def}^{\mathrm{ref}}_P(w)
+\operatorname{def}^{G}_P(w)
\Bigr\}.
\label{eq:graphroute}
\end{align}
In particular,
\begin{align}
\abs{w(P,\mathsf M)-w(P,\mathsf N)}
\leq\min\Bigl\{&
\bigl(c_P(\mathsf M)+c_P(\mathsf N)\bigr)
\operatorname{def}^{\mathrm{ref}}_P(w),\nonumber\\
&\bigl(s_P(\mathsf M)+s_P(\mathsf N)\bigr)
\operatorname{def}^{\mathrm{ref}}_P(w)
+\operatorname{def}^{G}_P(w)
\Bigr\}.
\label{eq:tworoute}
\end{align}
Consequently, with $m_P=\max\set{n_P-1,0}$,
\begin{equation}
\operatorname{def}^{G}_P(w)
\leq 2m_P\operatorname{def}^{\mathrm{ref}}_P(w)
\label{eq:defectrelation}
\end{equation}
and
\begin{equation}
\operatorname{osc}_P(w)
\leq
m_P\operatorname{def}^{\mathrm{ref}}_P(w)
+\frac12\operatorname{def}^{G}_P(w)
\leq
2m_P\operatorname{def}^{\mathrm{ref}}_P(w).
\label{eq:uniformtworoute}
\end{equation}
\end{theorem}

\begin{proof}
The case $P=I$ is immediate, so suppose $P\neq I$ and let $\mathsf B_P=\{P,I-P\}$ be the unique binary measurement containing $P$. Coarsen $\mathsf M$ to $\mathsf B_P$ by merging complementary outcomes. This is the reverse of a chain of $c_P(\mathsf M)$ elementary complementary refinements, so
\[
\abs{w(P,\mathsf M)-w(P,\mathsf B_P)}
\leq c_P(\mathsf M)\operatorname{def}^{\mathrm{ref}}_P(w).
\]
Doing the same for $\mathsf N$ and applying the triangle inequality gives the first route in~\eqref{eq:graphroute}.

For the second route, choose a shortest path
\[
\lambda_P(\mathsf M)=\lambda^{(0)},\lambda^{(1)},\ldots,\lambda^{(r)}=\lambda_P(\mathsf N)
\]
in $\Gamma_{n_P}$. Each edge can be realized on an actual measurement containing $P$: a split of a part is implemented by decomposing the corresponding outcome subspace into orthogonal subspaces of the prescribed positive dimensions, and a merge is the reverse operation. Starting from $\mathsf M$, follow this path to obtain a measurement $\widetilde{\mathsf N}$ with complementary rank profile $\lambda_P(\mathsf N)$. The triangle inequality gives
\[
\abs{w(P,\mathsf M)-w(P,\widetilde{\mathsf N})}
\leq r\operatorname{def}^{\mathrm{ref}}_P(w).
\]
By Theorem~\ref{thm:orbit}, there is a $U\in G_P$ such that $U\widetilde{\mathsf N}U^*=\mathsf N$, and therefore
\[
\abs{w(P,\widetilde{\mathsf N})-w(P,\mathsf N)}
\leq\operatorname{def}^{G}_P(w).
\]
This proves the second route in~\eqref{eq:graphroute}.

A path from each profile to $(1^{n_P})$ shows that
\[
d_{n_P}\bigl(\lambda_P(\mathsf M),\lambda_P(\mathsf N)\bigr)
\leq s_P(\mathsf M)+s_P(\mathsf N),
\]
which gives~\eqref{eq:tworoute}. To prove~\eqref{eq:defectrelation}, take $\mathsf N=U\mathsf M U^*$ in the first route. Then $c_P(\mathsf N)=c_P(\mathsf M)\leq m_P$; taking the supremum over $\mathsf M$ and $U$ yields the claim.

Finally, put $a=c_P(\mathsf M)+c_P(\mathsf N)$. Since
\[
s_P(\mathsf M)+s_P(\mathsf N)=2m_P-a,
\]
Equation~\eqref{eq:tworoute} gives
\[
\abs{w(P,\mathsf M)-w(P,\mathsf N)}
\leq
\min\bigl\{a\delta,(2m_P-a)\delta+\gamma\bigr\},
\]
where $\delta=\operatorname{def}^{\mathrm{ref}}_P(w)$ and $\gamma=\operatorname{def}^{G}_P(w)$. The minimum is no greater than the arithmetic mean of its two entries, namely $m_P\delta+\gamma/2$. Taking the supremum proves the first inequality in~\eqref{eq:uniformtworoute}; the second follows from~\eqref{eq:defectrelation}.
\end{proof}

Theorem~\ref{thm:refinement} shows that complementary-unitary invariance is not an additional exact requirement on the full measurement domain. The exact logical relation is
\[
\text{refinement consistency}
\quad\Longleftrightarrow\quad
\text{context independence}
\quad\Longrightarrow\quad
\text{complementary-unitary invariance},
\]
and the final implication is strict. Example~\ref{ex:uniform} is complementary-unitary invariant but not context independent and not refinement consistent. Theorem~\ref{thm:bridge} explains why complementary-unitary control remains useful at finite error: the orbit route may be shorter than coarsening through the binary context.

\section{The Gleason consequence}

For completeness, recall the finite-dimensional frame-function form of Gleason's theorem: if $d\geq3$ and $f$ assigns a number in $[0,1]$ to each rank-one projection such that
\[
\sum_{P\in\mathsf M}f(P)=1
\qquad\text{for every }\mathsf M\in\MM_1(\HH),
\]
then there is a unique density operator $\rho$ on $\HH$ such that $f(P)=\Tr(\rho P)$ for every rank-one $P$~\cite{Gleason1957,MorettiPastorello2013}. No countable-additivity hypothesis is needed in this finite-dimensional formulation.

\begin{theorem}[Invariant weights and the Born form]\label{thm:gleason}
Let $d\geq3$, and let $w$ be a context-indexed probability assignment on $\MM_1(\HH)$. The following are equivalent:
\begin{enumerate}
\item $w$ is complementary-unitary invariant;
\item $w$ is context independent;
\item there is a unique density operator $\rho$ on $\HH$ such that
\[
w(P,\mathsf M)=\Tr(\rho P)
\]
for every $\mathsf M\in\MM_1(\HH)$ and every $P\in\mathsf M$.
\end{enumerate}
The first condition may equivalently be replaced by invariance under two-level complementary unitaries.
\end{theorem}

\begin{proof}
The equivalence of (1) and (2) is the maximal-measurement case of Theorem~\ref{thm:profile}; Theorem~\ref{thm:twolevel} gives the stated local replacement. If (2) holds, define $f(P)$ to be the common value of $w(P,\mathsf M)$ over maximal measurements containing $P$. Normalization of $w$ implies that $f$ is a frame function of weight one. Gleason's theorem therefore gives a unique density operator $\rho$ satisfying $f(P)=\Tr(\rho P)$. This proves (3). Finally, an assignment of the form in (3) is normalized on every maximal projective measurement and depends only on $P$, so it satisfies both (1) and (2).
\end{proof}

\begin{theorem}[Born form on all projective measurements]\label{thm:allgleason}
Let $d\geq3$, and let $w$ be a normalized context-indexed probability assignment on $\MM(\HH)$. The following are equivalent:
\begin{enumerate}
\item $w$ is refinement consistent;
\item $w$ is context independent;
\item there is a unique density operator $\rho$ such that
\[
w(P,\mathsf M)=\Tr(\rho P)
\]
for every projective measurement $\mathsf M$ and every $P\in\mathsf M$.
\end{enumerate}
\end{theorem}

\begin{proof}
The equivalence of (1) and (2) is Theorem~\ref{thm:refinement}. Under either condition, let $\mu(P)$ be the common value assigned to the nonzero projection $P$, and set $\mu(0)=0$. Lemma~\ref{lem:additivity}, together with normalization, makes $\mu$ a finitely additive probability measure on the projection lattice. Gleason's theorem gives a unique density operator $\rho$ with $\mu(P)=\Tr(\rho P)$ for every projection $P$. This proves (3). Conversely, the Born assignment is normalized, context independent, and unchanged when an outcome other than $P$ is refined.
\end{proof}

\begin{remark}[Dimension two]\label{rem:d2}
When $d=2$, each rank-one $P$ lies in the unique maximal measurement $\set{P,\id-P}$. On maximal measurements, complementary-unitary invariance and context independence are therefore automatic. On all projective measurements, refinement consistency is also vacuous: the only elementary refinement is the split of $\set{I}$ into $\set{P,I-P}$, which leaves no outcome unaffected. Context independence is automatic for the same reason that each nontrivial outcome occurs in only one measurement. Any function $f$ on rank-one projections satisfying
\[
f(\id-P)=1-f(P)
\]
defines a normalized assignment on maximal measurements, and most such functions are not of the form $\Tr(\rho P)$. Thus the orbit and refinement theorems remain valid in dimension two, while the Born-form conclusions begin in dimension three. Gleason-type results for positive-operator-valued measurements can recover the qubit case under a larger measurement domain~\cite{Busch2003,CavesEtAl2004}.
\end{remark}

\section{An explicit symmetry-breaking assignment}

The preceding results are exact, but the broken symmetry can also be seen directly in a smooth normalized example. Let $\sigma$ be a density operator and let $\alpha>0$. On maximal measurements define
\begin{equation}\label{eq:power}
w_{\sigma,\alpha}(P,\mathsf M)
=\frac{\bigl(\Tr(\sigma P)\bigr)^\alpha}
{\displaystyle\sum_{Q\in\mathsf M}\bigl(\Tr(\sigma Q)\bigr)^\alpha}.
\end{equation}
The denominator is positive, so this is a well-defined normalized assignment. At $\alpha=1$, the denominator is $\Tr\sigma=1$ and~\eqref{eq:power} is the Born assignment. For other exponents, normalization generally introduces dependence on the other outcomes in the measurement.

Take $\HH=\C^3$ with orthonormal basis $e_1,e_2,e_3$, and put
\[
\psi=\frac{e_1+e_2}{\sqrt2},\qquad
\sigma=|\psi\rangle\langle\psi|,
\qquad P_1=|e_1\rangle\langle e_1|.
\]
For the standard maximal measurement
\[
\mathsf M_0=\set{P_1,P_2,P_3},
\]
the Born weights are $(1/2,1/2,0)$. With $\alpha=2$, equation~\eqref{eq:power} gives
\[
w_{\sigma,2}(P_1,\mathsf M_0)
=\frac{(1/2)^2}{(1/2)^2+(1/2)^2}=\frac12.
\]
Now set
\[
f_2=\frac{e_2+e_3}{\sqrt2},\qquad
f_3=\frac{e_2-e_3}{\sqrt2},
\]
and let
\[
\mathsf M_1=\set{P_1,|f_2\rangle\langle f_2|,|f_3\rangle\langle f_3|}.
\]
There is a unitary $U\in G_{P_1}$ with $Ue_2=f_2$ and $Ue_3=f_3$, so $\mathsf M_1=U\mathsf M_0U^*$. The Born weights in $\mathsf M_1$ are $(1/2,1/4,1/4)$, and therefore
\[
w_{\sigma,2}(P_1,\mathsf M_1)
=\frac{(1/2)^2}{(1/2)^2+(1/4)^2+(1/4)^2}=\frac23.
\]
Hence
\[
\abs{w_{\sigma,2}(P_1,\mathsf M_1)-w_{\sigma,2}(P_1,\mathsf M_0)}=\frac16.
\]
The rule is normalized and continuous in the measurement, yet it changes when only the orthogonal decomposition of $P_1^\perp$ is rotated. By Proposition~\ref{prop:defect}, $1/6$ is a lower bound for both its contextual oscillation and its complementary-unitary defect on that orbit.

\section{Scope and interpretation}

The mathematical conclusions may be summarized as follows. First, the complementary rank profile is a complete invariant for the action of $G_P$ on projective measurements containing $P$, and each orbit is an explicit compact homogeneous space. Second, on maximal rank-one measurements the profile is fixed, so a probability assignment is independent of the surrounding measurement exactly when it is invariant under complementary unitaries; two-level complementary unitaries already generate the required test. Third, on all projective measurements, refinement consistency alone is equivalent to context independence because every context through $P$ is linked to the unique binary context $\{P,I-P\}$. Fourth, the finite orbit quotient is the rank-profile refinement graph $\Gamma_{n_P}$, and Theorem~\ref{thm:bridge} bounds context variation using either binary coarsening or a shortest graph path followed by one unitary orbit move. Finally, in dimension at least three, Gleason's theorem converts the corresponding exact conditions into the Born form.

The results do not infer context independence from a separate dynamical, operational, or decision-theoretic premise. On maximal measurements, one exact condition is expressed as a group invariance. On the full projective-measurement domain, refinement consistency is already strong enough to remove both profile and orientation dependence: one may coarsen all the way to $I-P$ and refine again in a different orientation. Complementary-unitary control remains useful quantitatively because it may provide a shorter path when the refinement defect is nonzero.

The complementary-unitary invariance used here should also be distinguished from environment-assisted invariance~\cite{Zurek2003}, in which a unitary acting on a system is compensated by a unitary acting on an entangled environment. The present condition involves a single system, no entanglement, and unitaries acting only on the orthogonal complement of a fixed outcome within one projective measurement.

Nor should the terminology be conflated with every use of \emph{contextuality} in quantum foundations. Kochen--Specker contextuality concerns the impossibility of certain deterministic, context-independent value assignments in dimension at least three~\cite{KochenSpecker1967}; generalized contextuality concerns ontological representations of operational equivalences~\cite{Spekkens2005}. The present analysis concerns only the narrower consistency condition for scalar probabilities assigned to repeated projection outcomes. Reviews of the broader landscape include~\cite{BudroniEtAl2022}.

For positive-operator-valued measurements, the orbit structure is richer: ranks no longer classify decompositions, and additional operator data enter. Extending the orbit and refinement analysis to that setting is a separate problem.

\begingroup
\footnotesize
\setlength{\itemsep}{0em}

\endgroup

\end{document}